\documentclass[a4paper]{article}
\usepackage{RRA4}
\RRNo{6174}
\usepackage{amsmath,amssymb,amsthm,verbatim,graphics,subfigure}
\usepackage[lined,algoruled,noend,algo2e]{algorithm2e}
\usepackage[usenames]{color}
\usepackage{hyperref}
\usepackage{type1cm}

\usepackage{graphicx}
\usepackage{multicol}
\usepackage[latin1]{inputenc}
\newtheorem{Lem}{Lemma}
\newtheorem{Thm}{Theorem}

\newcommand{\mA}{\mathcal{A}}

\newcommand{\mG}{\mathcal{G}}

\newcommand{\mS}{\mathcal{S}}

\newcommand{\QQ}{\mathbb{Q}}

\renewcommand{\paragraph}[1]{\vspace{.1cm}\noindent\emph{#1~~}}

\RRdate{April 2007}
\RRauthor{Anh-Tuan Gai\thanks[INRIA]{INRIA, domaine de Voluceaux, 78153 Le Chesnay cedex, France} \and Dmitry Lebedev\thanks[FT]{France Telecom R\&D, 38--40, rue du g\'{e}n\'{e}ral Leclerc, 92130 Issy les Moulineaux, France} \and Fabien Mathieu\thanksref{FT} \and Fabien de Montgolfier\thanks{LIAFA, 175, rue du Chevaleret, 75013 Paris France} \and Julien Reynier\thanks{Motorola Labs, Parc les Algorithmes, Saint Aubin, 91193 Gif sur Yvette Cedex, France} \and Laurent Viennot\thanksref{INRIA}}
\RRtitle{Systèmes de préférences acycliques dans les réseaux pair-à-pair}
\RRetitle{Acyclic preference systems in P2P networks}

\titlehead{Acyclic preference systems in P2P networks}

\RRnote{CRC MARDI}
\RRabstract{
In this work we study preference systems natural for the Peer-to-Peer paradigm. Most of them fall in three categories: global, symmetric and complementary. All these systems share an acyclicity property. As a consequence, they admit a stable (or Pareto efficient) configuration, where no participant can collaborate with better partners than their current ones.

We analyze the representation of the such preference systems and show that any acyclic system can be represented with a symmetric mark matrix. This gives a method to merge acyclic preference systems and retain the acyclicity. We also consider such properties of the corresponding collaboration graph, as clustering coefficient and diameter. In particular, studying the example of preferences based on real latency measurements, we observe that its stable configuration is a small-world graph. 
}

\RRkeyword{P2P, stable marriage theory, rational choice theory, collaborative systems, BitTorrent, overlay network, matchings, graph theory}

\RRresume{
Cet article est consacré à des systèmes de préférences naturels dans les réseaux pair-à-pair. La plupart de ces systèmes appartiennent à l'une des trois classes de préférences suivantes : globale, symétrique et complémentaire. Ces préférences ont la particularité d'être acycliques. En conséquence, ils possèdent une unique configuration stable (au sens de Pareto), où aucun participant ne peut collaborer avec de meilleurs partenaires que ses partenaires courants.

En analysant leur représentation, nous montrons que tout système de préférences acycliques peut être représenté par une matrice de notes symétriques. Nous obtenons ainsi une méthode pour fusionner des systèmes de préférences acycliques en conservant la propriété d'acyclicité. Nous étudions également des propriétés du graphe de collaboration de la configuration stable, comme le diamètre ou le coefficient de clustering. Sur un example réel de préférences basées sur des mesures de latences, nous observons que la configuration stable a des caractéristiques petit-monde.}

\RRmotcle{pair-à-pair, mariages stables, choix rationnels, systèmes collaboratifs, BitTorrent, réseaux overlay, couplages, graphes}%%
\RRprojets{Gyroweb}
\RRtheme{\THCom}% \THCog \THSym \THNum \THBio} % cas de 5 themes
\URRocq % pour ceux qui sont au Sud.
\begin{document}
\makeRR   % cas d'un rapport de recherche

%\keywords{peer-to-peer, P2P, matching theory, preference list, stability, Pareto efficiency}
%\todo{DL: remove a page of text, chatterboxes :)) }

\section{Introduction}
\label{subsec:p2pnetworks}

\paragraph{Motivation} In most current peer-to-peer (P2P) solutions
participants are encouraged to cooperate with each other.
%Clients have some knowledge about other participants and the number of connections, that
Since collaborations may be costly in terms of network resources
(connection establishment, resource consumption, maintenance), the
number of connections is often bounded by the protocol. This
constraint encourages the clients to make a careful choice among others to
obtain a good performance from the system. The possibility to choose a
better partner implies that there exists \emph{a preference system}, which
describes the interests of each peer.

The study of such preference systems is the subject of $b$-matching
theory. It has started forty-five years ago with the seminal work of
Gale and Shapley on \emph{stable
marriages}~\cite{gale62college}. Although the original paper had a
certain \emph{recreational mathematics} flavor, the model turned out
to be especially valuable both in theory and practice. Today,
$b$-matching's applications are not limited to dating agencies, but
include college admissions, roommates attributions, assignment of
graduating medical students to their first hospital appointments, or
kidney exchanges
programs~\cite{gale62college,irving00hospitals,irving02stable,roth05pairwise}.
The goal of the present paper is to expand $b$-matching application domain
to P2P networks by using it to model the interactions between the
clients of such networks.

%\todo{DL: IMHO: je dirais que il y a un peu trop de paragraph
%commencant en ``emph''. C'est pas mal de souligner les contributions pour que
%les rapporteurs les voient en clair, mais les autres paragraphes nuisent un peu a la
%fluidité de la lecture - FdM moi je trouve ca OK... }

 \paragraph{Previous work} In \cite{inoc} we covered
general aspects of the $b$-matching theory application to the dynamics
of the node interactions. We considered preference systems natural for
the P2P paradigm, and showed that most of them fall into three
categories: global, symmetric, and complementary. We demonstrated that
these systems share the same property: acyclicity. We proved existence
and uniqueness of a stable configuration for acyclic preference
systems.

%\todo{DL: local marks n'ont pas encore definis, je suis pas sur si
%c'est comprehensible mais je vois pas comment l'ameilleurer pour
%l'instant }
\paragraph{Contribution} In this article, we analyze
the links between properties of local marks and the preference lists
that are generated with those marks. We show that all acyclic systems
can be created with symmetric marks. We provide a method to merge any
two acyclic preference systems and retain the acyclic property. And
finally our simulations show that real latency marks create
collaboration graphs with small-worlds properties, in contrast with
random symmetric or global marks.

 %\todo{DL:slightly too general and a copy of the abstract, to do more work }
%Roadmap
%\todo{DL cf mon commentaires sur emph: J'aurais mis le standard ``This article is structured as follows.'' Roadmap me semble de signifier des gros projet lié au temps}
\paragraph{Roadmap} In Section~\ref{sec:model} we define the
global, symmetric, complementary, and acyclic preference systems, and
provide a formal description of our model.  In
Section~\ref{subsec:equiv} we demonstrate that all acyclic preferences
can be represented using symmetric preferences. We consider
complementary preferences in Section~\ref{sec:prop}, and the results
are extended to any linear combination of global or symmetric
systems. Section~\ref{sec:smallworlds} discusses the
properties of a stable solution providing an example based on Meridian
project measurements \cite{meridianproject}. In Section~\ref{disc} we discuss the impact of our results, and Section~\ref{sec:conclusion} concludes.

%\todo{
%LV: o\`{u} est la section \ref{sectglobal} ? Pourquoi pas d\'{e}finir formellement les
%prefs globales en 3.3 aussi ?\\
%FM: je m'en occupe, c'est pr\'{e}vu.
%}

\section{Definition and applications of P2P preference systems}
\label{sec:model}
\subsection{Definitions and general modeling assumptions}\label{sec:qq}

We formalize here a $b$-matching model for common P2P preference systems.\\

\paragraph{Acceptance graph} Peers may have a partial knowledge
of the network and are not necessarily aware of all other
participating nodes. Peers may also want to avoid collaboration with
certain others. Such criteria are represented by an \emph{acceptance} graph
$G(V,E)$. Neighbors of a peer $p\in V$ are the nodes that may
collaborate with $p$. A configuration $C$ is a subset
$\tilde{C}\subset E$ of the existing collaborations at a given time.\\

%We suppose that at every time step $i$
%some peer $p$ is given \emph{an initiative} and it can decided what neighbors

\paragraph{Marks} We assume peers use some real marks (like latency, bandwidth,\ldots) to rank their neighbors. This is represented by a  valued matrix of marks $m=\{m(i,j)\}$. A peer $p$ uses $m(p,i)$ and $m(p,j)$ to compare $i$ and $j$. Without loss of generality, we assume that $0$ is the best mark and $m(p,i) < m(p,j)$ if and only if $p$ prefers $i$ to $j$. If $p$ is not a neighbor of $q$, then $m(p,q)=\infty$. We assume for convenience a peer $p$ has a different mark for each of its neighbors. It implies that a peer can always compare two neighbors
and decide which one suits better to him.\\%  Of course, a same instance $L$
%can be represented by different matrices.

\paragraph{Preference system} A mark matrix $M$ creates an instance $L$ of a preference system. $L(p)$ is a preference list that indicates how a peer $p$ ranks its neighbors. The relation when $p$ prefers $q_1$ to $q_2$ is denoted by $L(p,q_1)<L(p,q_2)$. Note that different mark matrices can produce the same preference system.\\

\paragraph{Global preferences}
A preference system is global if it can be deduced from global marks ($m(i,p) = m(j,p) = m(p)$).\\

\paragraph{Symmetric preferences} 
A preferences system is symmetric if it can be deduced from symmetric marks ($m(i,j)=m(j,i)$ for all $i,j$).\\%. This is equivalent to weighted graphs with the restriction that equal weights are not allowed in a same line.

\paragraph{Complementary preferences} 
A preferences system is complementary if it can be deduced from marks of the form $m(i,j)=v(j)-c(i,j)$, where $v(j)$ values the resources possessed by $j$ and $c(i,j)$ the resources that $i$ and $j$ have in common\footnote{Of course, in this case the preferred neighbor has a larger mark.}.\\%. This is equivalent to weighted graphs with the restriction that equal weights are not allowed in a same line.

\paragraph{Acyclic preferences} 
A preferences system is acyclic if it contains no preference cycle. A preference cycle is a cycle of at least three peers such that eacu peer strictly prefers its successor to its predecessor along the cycle.\\

\paragraph{Quotas} Each peer $p$ has a quota $b(p)$ (possibly infinite) on the number of links it can support. A \emph{$b$-matching} is a configuration $C$ that respects the quotas. If the quotas are greater than the number of possible collaborations, then simply $C = E$ would be an optimal solution for all.\\

\paragraph{Blocking pairs} We assume that the nodes aim to improve their situation, i.e. to link to most preferred neighbors. A pair of neighbors $p$ and $q$ is \emph{a blocking pair} of a configuration
$C$ if $\{p,q\}\in E\setminus C$ and both prefer to change the configuration $C$ and to link with the each other.  We assume that system evolves by discrete steps. At any step two nodes can be linked together if and only if they form a blocking pair. Those nodes may drop their worst performing links to stay within their quotas. A configuration $C$ is \emph{stable} if no blocking pairs exist.
%\todo{DL:more words on stability? pareto efficiency? FM: nan, pas la place et pas le temps.}\\

%%%%% PLUS TARD, Ã§a
\paragraph{Loving pair} Peers $p,q$ form a loving pair if $p$ prefers $q$ to all its
other neighbors and $q$, in its turn, prefers $p$ to all other neighbors. It implies a strong link which cannot be destroyed in the given preference system.

%%%%%%%%%%%%%%%%%%%%%%%%%%%%%%%%%%%%%%%%%%%%%%

\subsection{Preference systems and application design}

Depending on the P2P application, several important criterion can be
used by a node to choose its collaborators. We introduce the following
three types as representative of most situations:

\paragraph{Proximity:} distances in the physical network, in a virtual space or similarities according to some characteristics,

\paragraph{Capacity related:} network bandwidth, computing capacity, storage capacity,

\paragraph{Distinction:} complementary character of resources owned by
different peers.

Notice that theses types correspond respectively to the definitions of symmetric, global and complementary preference system categories.

Examples of symmetric preferences are P2P applications which optimize latencies.  A classical approach for distributed hash-table lists of contacts is selecting the contacts with the smallest round trip time (RTT) in the physical network. In Pastry~\cite{pastry}, a node will always prefer contacts with smallest RTT among all the contacts that can fit into a given routing table entry.  More generally, building a low latency overlay network with bounded degree requires to select neighbors with small RTTs.  Optimizing latencies between
players can also be crucial, for instance, for online real-time gaming
applications~\cite{syncms}. Such preferences are symmetric since the mark a peer $p$ gives to some peer $q$ is the same as the mark $q$ gives to $p$ (the RTT between $p$ and $q$).

Similarly, massively multiplayer online games (MMOG) require
connecting players with nearby coordinates in a virtual space~\cite{keller03simon,mmog}.  Again this can be modeled by symmetric
preferences based on the distance in the virtual space. Some authors also propose to connect participants of a file sharing system
according to the similarity of their interests~\cite{fessant04clustering,sripanidkulchai03efficient}, which
is also a symmetric relation.

%Another symmetric preference system is mutual trust: each time a
%successful transaction between two peers is performed (exchanging data
%in a filesharing system, playing a game online without abusive
%resignation, and so on) the mutual trust is increased. This is
%important for P2P applications where security issues matter.
%
%
%\todo{FM: On a une rÃ©f pour mutual trust? J'ai rÃ©cupÃ©rÃ© Ã§a chez FdM.}

%% BitTorrent~\cite{cohen03incentives} is an example of
%% P2P application that uses a capacity related preference system.  Shortly and with many simplifications the protocol can be described as
%% follows. A set of peers try to exchange pieces of a given file to obtain the whole file. Each peer $p$ divides its upload bandwidth into a certain number of slots. One of those slot distributes generously pieces of the file to random peers. The other slots give preferentially to the peers that upload the best to $p$.
%% This ``Tit-for-Tat'' strategy tends to create symmetric
%% exchanges between pairs of nodes~\cite{izal04dissecting}. And obviously,
%% each peer prefers to connect to peers with best upload capacity per slot.  This implies that the nodes are ranked according to their capacity per slot, therefore, the mechanism can be considered as a global preference system.
BitTorrent~\cite{cohen03incentives} is an example of a P2P application
that uses a capacity related preference system. In brief, a BitTorrent peer uploads to peers it has most downloaded from during the last ten seconds. This is an implementation of the well known Tit-for-Tat strategy. The mark of a peer can thus be seen has its upload capacity divided by its collaboration quota.

%In brief the protocol can be described as follow. BitTorrent gather the set of peers $\cal{P}$ interesting in downloading the same file. Each peer is connected to a subset of peers in $\cal{P}$ with which it regularly exchange the list of file blocks it possess. 

% In brief the protocol
%can be described as follows. Its clients are interested in downloading
%a large file. To minimize the load on a source, the file is divided in
%many parts, which are downloaded by different nodes. Then, the nodes
%connect to each other and exchange these parts. Every node has a bound
%on the number of upload links. The protocol defines ``Tit-for-Tat''
%strategy on the upload, which is made reciprocal to the download for
%every link. This strategy tends to create symmetric traffic between
%the nodes~\cite{izal04dissecting}. Clearly, each peer prefers to
%connect to peers with the best upload capacity.  This implies that the
%nodes are ranked according to their capacity, therefore, the mechanism
%can be considered as a global preference system. It should be noted,
%that the protocol reserves one link for a generous upload, which
%allows to probe random nodes, in case, that performance of their links
%improves.
%The generous connection can be seen as a probing mechanism for finding mates with better upload capacity.  

This global preference nature of BitTorrent should be tempered by the
fact that only peers with complementary parts of the file are
selected. Pushing forward this requirement would lead to another
selection criterion for BitTorrent: preference for the peers possessing the
most complementary set of file pieces. In other words, each peer
should try to exchange with peers possessing a large number of blocks
it needs. We call this a complementary preference system. Note, that
this kind of preferences changes continuously as new pieces are
downloaded.  However, the peers with the most complementary set of blocks
are those, who enable longest exchange sessions.

In its more general form, the selection of partners for cooperative file download can be seen as a mix of several global, symmetric, and complementary preference systems.

\section{Acyclic preferences equivalence}
\label{subsec:equiv}

In \cite{inoc}, we showed that global, symmetric and complementary preferences are acyclic, and that any acyclic $b$-matching preference instance has a unique stable configuration. And acyclic systems always converge toward their stable configuration. However, since acyclicity is not defined by construction, one can ask if other kind of acyclic preferences exist. This section is devoted to answer this question.

%Definition \ref{def:acyclic} of acyclic preference systems using
%Kieschick cycles is not constructive. It is more practical to find an
%equivalent class of preference systems defined by their marks.  It
%appears that symmetric preferences can serve exactly for this purpose.

%In this section we demonstrate that all acyclic preference
%systems can be expressed in terms of symmetric systems. 

\begin{Thm}\label{grandTh}
Let $P$ be a set of $n$ peers, $\mA$ be the set of all possible
acyclic preference instances on $P$, $\mS$ be the set of all possible
symmetric preference instances on $P$, $\mG$ be the set of all possible
global preference instances, then

$$\mG\subsetneqq\mA = \mS $$
\end{Thm}
The rest of this section consists of the proof of theorem
\ref{grandTh}: we will first show $\mS\subset \mA$ and $\mG\subset \mA$, then $\mA
\subset \mS$, which will be followed by $\mG\neq\mA$.
\begin{Lem}
\label{thm:glob_acycl}
Global and symmetric preference systems are acyclic
\end{Lem}

\begin{proof} \emph{from~\cite{inoc}}
Let us assume the contrary, and assume that there is a circular list of peers $p_1,\ldots,p_k$ (with $k\geq 3$), such that each peer of the list strictly prefers its successor to its predecessor. Written in the form of marks it means that $m(p_i,p_{i+1})<m(p_i,p_{i-1})$ for all $i$ modulo $k$. Taking a sum for all possible $i$, we get $$\sum_{i=1}^k m(p_i,p_{i+1})<\sum_{i=1}^k
m(p_i,p_{i-1})\text{.}$$
\noindent If marks are global, this can be rewritten $\sum_{i=1}^k m(p_i)<\sum_{i=1}^k
m(p_i)$, and if they are symmetric, $\sum_{i=1}^k m(p_i,p_{i+1})<\sum_{i=1}^k
m(p_i,p_{i+1})$. Both are impossible, thus global and symmetric marks create acyclic instances.\qed
 
\end{proof}

The next part, $\mA\subset \mS$, uses the loving pairs described in \ref{sec:qq}. We first prove the existence of loving pairs in Lemma~\ref{lemGai}.

\begin{Lem}\label{lemGai}
A nontrivial acyclic preference instance always admits at least one loving pair.
\end{Lem}
\begin{proof}
A formal proof was presented in \cite{inoc}. In short, if there is no loving pair, one can construct a preference cycle by considering a sequence of first choices of peers.
\end{proof}

\begin{algorithm2e}[htb]
\caption{Construction of a symmetric note matrix $m$  given an acyclic
preferences instance $L$ on $n$ peers  \label{algo:notes}} \dontprintsemicolon 
%\KwIn {A preference instance $L$ on $n$ peers} 
%\KwOut{A symmetric marks matrix $m$}
%\Begin{
$N := 0$\;
for all $p$ and $q$, $m(p,q)=+\infty$ (by default, peers do not accept each other)\;
\While{there exists a loving pair $\{a,b\}$}{
  $m(a,b):=m(b,a):=N$\;
  Remove $a$ from the preference list $L(b)$ and $b$ from  $L(a)$\;
  $N:=N+1$\;
}
%Note: if there remains nonempty preferences list then the preferences set is not acyclic\;
%}
\end{algorithm2e}

\begin{Lem}
Let $L$ be a preference instance. Algorithm~\ref{algo:notes}
constructs a symmetric mark matrix in $O(n^2)$ time that produces
$L$.
\end{Lem}
\begin{proof}
The matrix output is clearly symmetric. Neighboring peers get finite
marks, while others have infinite marks.  If an instance contains a
loving pair $\{a,b\}$ then $m(a,b)=m(b,a)$ can be the best mark since
$a$ and $b$ mutually prefer to any other peers.  According to Lemma~
\ref{lemGai} such a loving pair always exists in acyclic case.  By
removing the peers $a$ and $b$ from their preferences lists, we are
lead to a smaller acyclic instance with the same preference lists except $a$ and $b$ are now unacceptable to each other. The process continues until all  preference lists are eventually empty.
 The marks are given in increasing order, therefore when  $m(p,q)$ and $m(p,r)$ are finite,  $m(p,q)<m(p,r)$ iff the loving pair $\{p,q\}$ is formed before the loving pair $\{p,r\}$ iff $p$ prefers $q$ to $r$. 

The algorithm runs in  $O(n^2)$ time because an iteration of the \textbf{while} loop takes $O(1)$ time.
A loving pair can especially be found in constant time by maintaining a list on all loving pairs. The list is updated in constant time since, after
$a$ and $b$ became mutually unacceptable, each new loving pair
contains either $a$ and its new first choice, or $b$  and its new first
choice.\qed
\end{proof}

\paragraph{All acyclic preferences are not global preferences.}
A simple counter-example uses $4$ peers $p_1$, $p_2$, $p_3$ and $p_4$
with the following preference lists:\\
$
L(p_1):\ p_2,p_3,p_4\ \ \ L(p_2):\ p_1,p_3,p_4\ \ \ 
L(p_3):\ p_4,p_1,p_2\ \ \ 
L(p_4):\ p_3,p_1,p_2
$

$L$ is acyclic, but $p_1$ prefers $p_2$ to $p_3$ whereas $p_4$ prefers $p_3$ to $p_2$. $p_1$ and $p_4$ rate $p_2$ and $p_3$ differently, thus the instance is not global.

\section{Complementary and Composite Preference Systems}
\label{sec:prop}

Complementary preferences appear in systems where peers are equally interested in the resources they do not have yet. As said in Section \ref{sec:qq}, complementary preferences can be deduced from marks of the form  $m(p,q)=v(q)-c(p,q)$ (in this case, marks of higher values are preferred).

The expression of a complementary mark matrix $m$ shows it is a linear combination of previously discussed global and symmetric mark matrices: $m=v-c$, where $v$ defines a global preference system and $c$ defines a symmetric system.

Theorem~\ref{th:lin} shows that complementary marks, and more generally any linear combination of global or symmetric marks, produce acyclic preferences.

\begin{Thm}\label{th:lin}
Let $m_1$ and $m_2$ be global or symmetric marks.  Any linear
 combination of $\lambda m_1+\mu m_2$ is acyclic.
\end{Thm}
\begin{proof}
The proof is practically the same as for Lemma~\ref{thm:glob_acycl}. Let us suppose that the preference system induced by $m=\lambda m_1+\mu m_2$ contains a preference cycle $p_1,p_2,\ldots, p_k, p_{k+1}=p_1$, for $k\geq 3$. We assume wlog that $m_1$ is global, $m_2$ symmetric and that marks of higher values are preferred for $m$. Then $m(p_{i},p_{i+1}) > m(p_i,p_{i-1})$  for all $i$ modulo $k$. Taking a sum over all possible $i$, we get
\begin{eqnarray*}
\sum_{i=1}^k m(p_{i},p_{i+1}) > \sum_{i=1}^k m(p_{i},p_{i-1})
 & = & \sum_{i=1}^k \left( \lambda m_1(p_{i},p_{i-1})+\mu m_2(p_{i},p_{i-1})\right)\text{, but}\\
\sum_{i=1}^k \left( \lambda m_1(p_{i},p_{i-1})+\mu m_2(p_{i},p_{i-1})\right) & = & \lambda \sum_{i=1}^k  m_1(p_{i-1})+\mu \sum_{i=1}^k m_2(p_{i-1},p_{i})\\
 & = & \lambda \sum_{i=1}^k  m_1(p_{i+1})+\mu \sum_{i=1}^k m_2(p_{i},p_{i+1})\\ &=&\sum_{i=1}^k m(p_{i},p_{i+1})\text{.}\\
\end{eqnarray*}
This contradiction proves the Theorem.
\qed \end{proof}

Theorem~\ref{th:lin} leads to the question whether any linear
combination of acyclic preferences expressed by any kind of marks is also acyclic. The example bellow illustrates that in general it is not true:
{\small
$$ M_1 = \left( \begin{array}{ccc} 0&3&1 \\ 2&0&1 \\ 3&1&0
\end{array}\right) \mbox{ , \hfill} M_2 = \left( \begin{array}{ccc} 0&1&2
\\ 1&0&3 \\ 1&2&0 \end{array}\right) 
 \mbox{ , \hfill}
M_1+M_2  =  \left( \begin{array}{ccc} 0&4&3 \\ 3&0&4 \\ 4&3&0
\end{array}\right)
$$
}
The preference instance induced by $M_1 + M_2$ has the cycle $1,2,3$,
while both $M_1$ and $M_2$ are acyclic (both produce global preferences).

Note, that a linear combination of two preference system matrices can
give duplicates in the marks of a single node, which generates ties in preferences. Ties affect existence and uniqueness of a stable configuration, depending on how they are handled. If a peer prefers a new node to a current collaborator that has the same mark, existence is not guaranteed (but if a stable configuration exists, it is unique). If not, existence stands, but not uniqueness\footnote{Irving and Manlove have performed a rather complete study on ties~\cite{irving02stable}.}.
  
However, Theorem \ref{th:lin} provides in conjunction with Theorem~\ref{grandTh} a way of constructing a tie-less acyclic instance that can take into account several parameters of the network, as long as all produce acyclic preferences. The parameters can be converted into integer symmetric marks using Algorithm~\ref{algo:notes}. A linear combination using $\QQ$-independent scalars produces distinct acyclic marks.%\todo{DL:develop this ? FM: OK}

% This is a specific case and we have excluded such matrices from our definition of preference systems by assuming that any two neighbors can be compared and one would be better than other. In the case where this rule is weakened, peers need clear instructions whether they prefer a new node with a same mark to its current collaborator. In the case of a positive answer such instances simply might not have a stable solution, since the nodes will continue to change their partners. Otherwise, clearly the uniqueness of the stable solution is not guaranteed.

\section{Graph properties of stable configurations}
\label{sec:smallworlds}

Many protocols use preference systems that come from global and symmetric marks. Studying the properties of the stable configuration for such protocols may give information on the performances one can expect. In this Section, we study connectivity properties for three cases. Connectivity is extensively studied since Watts survey~\cite{SW03} on the \emph{small world} graphs. These graphs are  known to have good routing and robustness properties. They are characterized by a small (i.e. $O(\log(n)$) mean distances and high (i.e. $O(1)$) clustering. The \emph{clustering coefficient} is the probability for two vertices $x$ and $y$ to be linked, giving %\todo{FdM: ``sachant que'' = ``knowing that ?'' FM:giving}
that $x$ and $y$ have at least one common neighbor.

The cases we considered all involved a set of $n=2500$ peers, and differ from the marks: the first uses a global mark matrix\footnote{In absence of tie, all global marks are the same up to permutation},the second a random symmetric mark matrix, and the last a latency mark matrix from the Meridian Project~\cite{meridianproject}.

\begin{figure}[ht]
\begin{center}
\subfigure[Diameter]{\includegraphics[width=.45\textwidth]{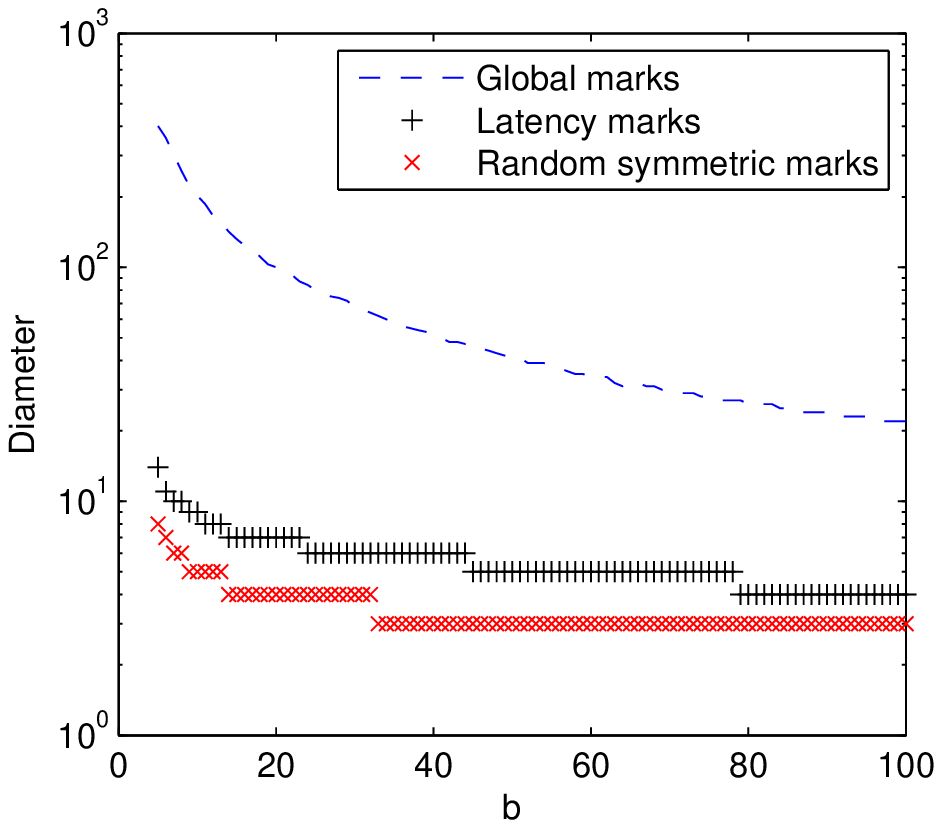}}\subfigure[Clustering coefficient]{\includegraphics[width=.45\textwidth]{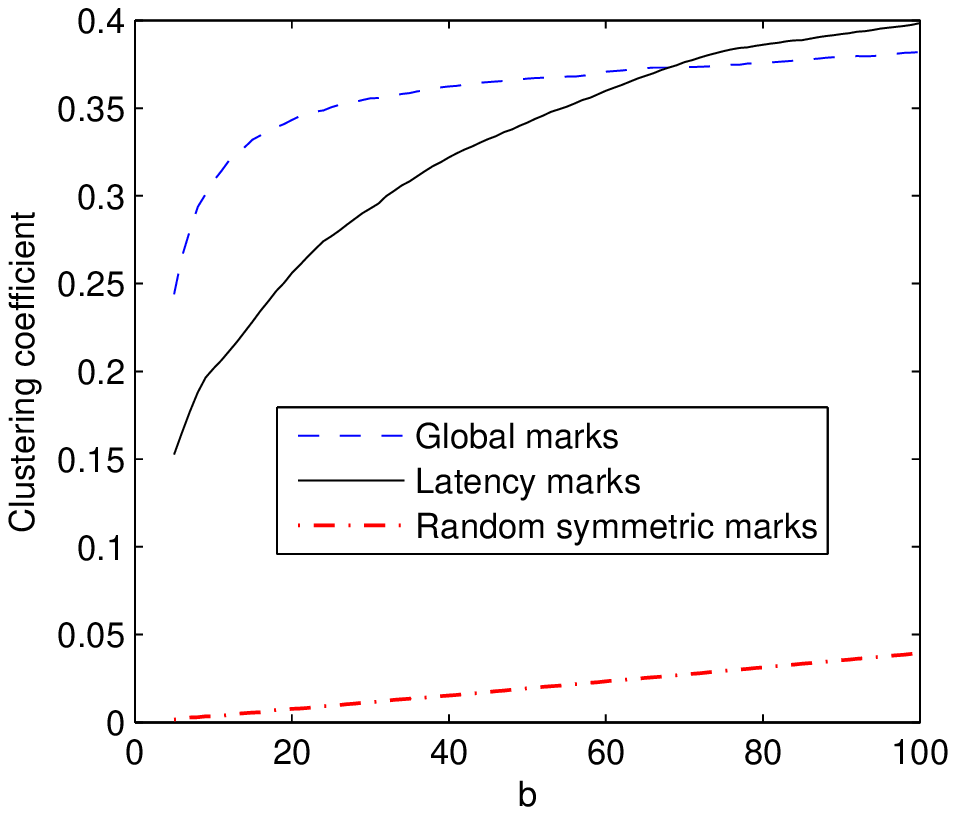}} \end{center}
\caption{Diameter and clustering coefficient of latency, random symmetric and global marks (2500 nodes) stable configurations. Global marks use an underlying Erdös-Rényi $\mG(2500,0.5)$ acceptance graph.}
\label{fig:mer1}
\end{figure}

Figure~\ref{fig:mer1} shows the properties of the stable configuration for these three marks, as a function of the quota $b$ on the number of links per peer.

Global marks produce configuration with disconnected cliques of size $b+1$ (maximal clustering, and infinite diameter). We had previously observed this \emph{clusterization} effect in~\cite{rr6081}. It can be lessened by using an Erdös-Rényi acceptance graph. Then the configuration still has a high clustering coefficient, and a high, but finite diameter (same order of magnitude than $\frac{n}{b}$). This is due to a \emph{stratification} effect: peers only link to peers that have marks similar to them~\cite{rr6081}.

Random symmetric matrix produces configurations with low diameter and clustering coefficient. There characteristics are similar to those of 
Erdös-Rényi graphs.

Real latencies produce both a low diameter and a high clustering coefficient. This indicates that the stable configuration has \emph{small-world} structure. But it is not a \emph{scale-free} network~\cite{newman03structure}, because the degree distribution does not follow a power law (the degrees are bounded by $b$).

\section{Discussion and future work}\label{disc}

\paragraph{Stability} Decision, whether a stable configuration is a good thing or not, depends on the characteristics and needs of practical applications. If continuous link alteration has a high cost (like in structured P2P networks), or if the stable configuration has appealing properties (like the small-world properties observed for latency-based stable configuration), then it is interesting to let the system converge. On the other hand, we observed that global marks result in a stable configuration with high diameter, which is an undesired feature in most cases. Moreover, some systems like gossip protocols\cite{allavena05correctness} take advantage of constant evolution of the corresponding acceptance graph.  In such cases, the eventual convergence would be harmful.

\paragraph{Convergence speed} The convergence speed is an important
characteristic, whether the stable solution is desired or not. In the
first case, the application is interested in speeding up the
process. In the second case, the slower possible speed is preferred
instead.  Although this question is out the scope of the present work,
our current experiments suggest that the convergence depends on many
parameters: the preference system used, the acceptance graph, the
activity of peers (details of peers' interaction protocol), the quotas
and others. If we use as time unit the mean interval between two
attempts of a given peer to change one of its neighbors, then
preliminary results show that convergence is logarithmic at best, and
polynomial at worst. We plan on providing a complete study on the
influence of parameters. This should help understanding existing
protocols and making them more efficient.

\paragraph{Dynamics of preference systems} 
We have considered fixed acceptance graph and preference lists. In real applications, arrivals and departures modify the acceptance graph, along with the discovery of new contacts (a toy example is BitTorrent, where a tracker periodically gives new contacts to the clients). 
The preference system itself can evolve in time. For instance, latency
can increase if a corresponding link has a congestion problem. A complementary preference system is dynamic by itself: as a peer gets resources from a complementary peer, the complementarity mark decreases.
 
All these changes impact the stable configuration of the system. The question is to know whether the convergence speed can sustain the dynamics of preferences or not. Fast convergence and slow changes allow the system to continuously adjust (or stay close) to the current stable configuration. Otherwise, the configurations of the system may be always far from a stable configuration that changes too often. The preferable behavior depends on whether stability is a good feature. This is an interesting direction for the future work.

\section{Conclusion}
\label{sec:conclusion}

In this article, we gave formal definitions for a $b$-matching P2P model and analyze the existence of a stable configuration with preference systems natural to P2P environment. The term stability in our case corresponds to Pareto efficiency of the collaboration network, since the participants have no incentives to change such links. We also showed that in contrast with systems based on intrinsic capacities, a latency-based stable configuration has small-world characteristics.
%\todo{DL:i've never seen a final summary. maybe, delete it away? but it seemed fancy to me.}

\bibliographystyle{plain}
% argument is your BibTeX string definitions and bibliography database(s)

\bibliography{maria}
 
%\pagebreak
% \todos

\end{document}